\documentclass[final,12pt]{article}
\usepackage{hyperref}
\usepackage{amssymb,amsmath,amsthm}
\usepackage{tikz}
\usetikzlibrary{arrows,backgrounds,decorations.pathmorphing,decorations.pathreplacing,positioning,fit}
\usepackage{enumerate}
\usepackage[conditional,light,first,bottomafter]{draftcopy}
\draftcopyName{DRAFT\space\today}{130}
\draftcopySetScale{65}
\usepackage[letterpaper,hmargin=3.7cm,vmargin=3.3cm]{geometry}
\usepackage{setspace}
\makeatletter
\renewcommand{\section}{\@startsection%
{section}%
{1}%
{0em}%
{1.7em}%
{1.2em}%
{\normalfont\large\centering\bfseries}}
\renewcommand{\@seccntformat}[1]%
{\csname the#1\endcsname.\hspace{0.5em}}
\makeatother
\renewcommand{\thesection}{\arabic{section}}

\numberwithin{equation}{section}
\renewcommand\appendix{\par
\setcounter{section}{0}%
\setcounter{subsection}{0}%
\setcounter{theorem}{0}
\setcounter{table}{0}
\setcounter{figure}{0}
\gdef\thetable{\Alph{table}}
\gdef\thefigure{\Alph{figure}}
\section*{Appendix}
\gdef\thesection{\Alph{section}}
\setcounter{section}{1}}

\newtheorem{theorem}{Theorem}[section]
\newtheorem{proposition}{Proposition}[section]
\newtheorem{lemma}{Lemma}[section]
\newtheorem{corollary}{Corollary}[section]
\theoremstyle{definition}
\newtheorem{definition}{Definition}
\newtheorem{remark}{Remark}

\newtheorem*{acknowledgments}{Acknowledgments}
\newcommand{\abs}[1]{\left|#1\right|}
\newcommand{\norm}[1]{\left\|#1\right\|}
\newcommand{\inner}[2]{\left\langle#1,#2\right\rangle}
\newcommand{\cc}[1]{\overline{#1}}
\newcommand{\reals}{\mathbb{R}}
\newcommand{\nats}{\mathbb{N}}
\newcommand{\complex}{\mathbb{C}}
\newcommand{\eval}[1]{\upharpoonright_{#1}}
\newcommand{\convergesto}[1]{\xrightarrow[#1\to\infty]{}}

\DeclareMathOperator{\dom}{dom}

\DeclareMathOperator{\Span}{span}
\begin{document}
\begin{titlepage}
  \title{Spectral analysis for linear semi-infinite\\
    mass-spring systems
\footnotetext{%
      Mathematics Subject Classification(2010):
47A75, 
47B36, 
}
\footnotetext{%
Keywords:
Jacobi matrices;
Spectrum;
Infinite mass-spring system
}
\hspace{-8mm}
\thanks{%
Research partially supported by UNAM-DGAPA-PAPIIT IN105414
}%
}
\author{
\textbf{Rafael del Rio and Luis O. Silva}
\\[6mm]
\small Departamento de F\'{i}sica Matem\'{a}tica\\[-1.6mm]
\small Instituto de Investigaciones en Matem\'aticas Aplicadas y en Sistemas\\[-1.6mm]
\small Universidad Nacional Aut\'onoma de M\'exico\\[-1.6mm]
\small C.P. 04510, M\'exico D.F.\\[1mm]
\small\texttt{delrio@iimas.unam.mx}\\[-1mm]
\small\texttt{silva@iimas.unam.mx}
}
\date{}
\maketitle
\vspace{-4mm}
\begin{center}
\begin{minipage}{5in}
  \centerline{{\bf Abstract}}\bigskip
  We study how the spectrum of a Jacobi operator changes when this
  operator is modified by a certain finite rank perturbation.  The
  operator corresponds to an infinite mass-spring system and the
  perturbation is obtained by modifying one interior mass and one
  spring of this system.  In particular, there are detailed results of
  what happens in the spectral gaps and which eigenvalues do not move
  under the modifications considered. These results were obtained by a
  new tecnique of comparative spectral analysis and they generalize
  and include previous results for finite and infinite Jacobi
  matrices.
\end{minipage}
\end{center}
\thispagestyle{empty}
\end{titlepage}
\section{Introduction}
\label{sec:intro}
Denote by $l_{\rm fin}(\mathbb{N})$ the linear space of complex
sequences having a finite number of non-zero elements. In the Hilbert
space $l_2(\mathbb{N})$, let $J_0$ be the operator with
$\dom(J_0)=l_{\rm fin}(\mathbb{N})$ such that, for every
$f=\{f_k\}_{k=1}^\infty$ in $l_{\rm fin}(\mathbb{N})$,
\begin{subequations}
\label{eq:main-recurrence}
\begin{align}
  \label{eq:initial-defining}
  (J_0f)_1&:= q_1 f_1 + b_1 f_2\,,\\
  \label{eq:recurrence-defining}
  (J_0f)_k&:= b_{k-1}f_{k-1} + q_k f_k + b_kf_{k+1}\,,
  \quad k \in \mathbb{N} \setminus \{1\},
\end{align}
\end{subequations}
where $q_n\in\mathbb{R}$ and $b_n>0$ for any $n\in\mathbb{N}$. The
operator $J_0$ is symmetric and therefore one can consider the operator
$\cc{J_0}$ being  its closure. For the symmetric operator $\cc{J_0}$,
one of the following two possibilities for the deficiency indices
holds \cite[Chap.\,4,\,Sec.\,1.2]{MR0184042}:
\begin{subequations}
  \label{eq:deficiency-indices}
  \begin{align}
     n_+(\cc{J_0})&=n_-(\cc{J_0})=1\,,\label{eq:one-one}\\
     n_+(\cc{J_0})&=n_-(\cc{J_0})=0\,.\label{eq:nil-nil}
  \end{align}
\end{subequations}
Fix a self-adjoint extension of $J_0$ and denote it by $J$. Thus, in
view of the possible values of the deficiency indices, the von Neumann
extension theory tells us that either $J$ is a proper closed symmetric
extension of $\cc{J_0}$ or $J=\overline{J_0}$. According to the
definition of the matrix representation for an unbounded symmetric
operator \cite[Sec. 47]{MR1255973}, $\overline{J_0}$ is the operator
whose matrix representation with respect to the canonical basis
$\{\delta_n\}_{n=1}^\infty$ in $l_2(\mathbb{N})$ is
\begin{equation}
  \label{eq:jm-0}
  \begin{pmatrix}
    q_1 & b_1 & 0  &  0  &  \cdots
\\[1mm] b_1 & q_2 & b_2 & 0 & \cdots \\[1mm]  0  &  b_2  & q_3  &
b_3 &  \\
0 & 0 & b_3 & q_4 & \ddots\\ \vdots & \vdots &  & \ddots
& \ddots
  \end{pmatrix}\,.
\end{equation}
The $k$-th entry of $\delta_n$ is 1 if $k=n$ and 0 if $k\ne n$.

Fix $n\in\nats$ and consider, along with the self-adjoint operator
$J$, the operator
\begin{equation}
\label{eq:def-tilde-j}
\begin{split}
  \widetilde{J}_n=J &+
  [q_n(\theta^2-1)+\theta^2h]\inner{\delta_n}{\cdot}\delta_n \\
  &+  b_n(\theta-1)(\inner{\delta_n}{\cdot}\delta_{n+1} +
  \inner{\delta_{n+1}}{\cdot}\delta_n) \\
  &+ b_{n-1}(\theta-1)(\inner{\delta_{n-1}}{\cdot}\delta_{n} +
  \inner{\delta_{n}}{\cdot}\delta_{n-1})
  \,,\quad \theta>0\,,
  \quad h\in\mathbb{R}\,,
\end{split}
\end{equation}
where it has been assumed that $b_0=0$. Clearly, $\widetilde{J}_n$ is a
self-adjoint extension of the operator whose matrix representation
with respect to the canonical basis in $l_2(\mathbb{N})$ is a Jacobi
matrix obtained from (\ref{eq:jm-0}) by modifying the entries
$b_{n-1},q_n,b_n$. For instance, if $n>2$, $\widetilde{J}_n$ is a
selfadjoint extension (possibly not proper) of the operator whose
matrix representation is
\begin{equation}
  \label{eq:jm-1}
  \begin{pmatrix}
q_1 & b_1 & 0 & 0 & 0 & 0 & \cdots \\[1mm]
b_1 & \ddots & \ddots & 0 & 0 & 0 & \cdots \\[1mm]
0  &  \ddots  & q_{n-1} & \theta b_{n-1} & 0 & 0 & \cdots\\
0 & 0 & \theta b_{n-1} & \theta^2(q_n+h) & \theta b_n & 0 & \cdots\\
0 & 0 & 0 & \theta b_{n} & q_{n+1} & \theta b_{n+1} & \\
0 & 0 & 0 & 0 & b_{n+1} & q_{n+2} & \ddots\\
\vdots & \vdots & \vdots & \vdots & & \ddots & \ddots
  \end{pmatrix}\,.
\end{equation}
Note that $\widetilde{J}_n$ is obtained from $J$ by a rank-three
perturbation when $n>1$, and by a rank-two perturbation when $n=1$.

The particular kind of perturbation given in (\ref{eq:def-tilde-j})
arises in the analysis of semi-infinite mass-spring systems. It is
known \cite{MR2998707,see-later-mr} that, within the regime of
validity of the Hooke law, the system in Fig.~1,
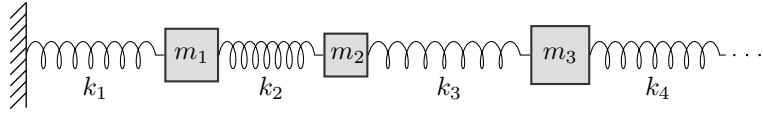
\begin{figure}[h]
\begin{center}
\begin{tikzpicture}
  [mass1/.style={rectangle,draw=black!80,fill=black!13,thick,inner sep=0pt,
   minimum size=7mm},
   mass2/.style={rectangle,draw=black!80,fill=black!13,thick,inner sep=0pt,
   minimum size=5.7mm},
   mass3/.style={rectangle,draw=black!80,fill=black!13,thick,inner sep=0pt,
   minimum size=7.7mm},
   wall/.style={postaction={draw,decorate,decoration={border,angle=-45,
   amplitude=0.3cm,segment length=1.5mm}}}]
  \node (mass3) at (7.1,1) [mass3] {\footnotesize$m_3$};
  \node (mass2) at (4.25,1) [mass2] {\footnotesize$\,m_2$};
  \node (mass1) at (2.2,1) [mass1] {\footnotesize$m_1$};
\draw[decorate,decoration={coil,aspect=0.4,segment
  length=2.1mm,amplitude=1.8mm}] (0,1) -- node[below=4pt]
{\footnotesize$k_1$} (mass1);
\draw[decorate,decoration={coil,aspect=0.4,segment
  length=1.5mm,amplitude=1.8mm}] (mass1) -- node[below=4pt]
{\footnotesize$k_2$} (mass2);
\draw[decorate,decoration={coil,aspect=0.4,segment
  length=2.5mm,amplitude=1.8mm}] (mass2) -- node[below=4pt]
{\footnotesize$k_3$} (mass3);
\draw[decorate,decoration={coil,aspect=0.4,segment
  length=2.1mm,amplitude=1.8mm}] (mass3) -- node[below=4pt]
{\footnotesize$k_4$} (9.3,1);
\draw[line width=.8pt,loosely dotted] (9.4,1) -- (9.8,1);
\draw[line width=.5pt,wall](0,1.7)--(0,0.3);
\end{tikzpicture}
\end{center}
\caption{Semi-infinite mass-spring system}\label{fig:1}
\end{figure}
with masses
$\{m_j\}_{j=1}^\infty$ and spring constants $\{k_j\}_{j=1}^\infty$, is
modeled by the Jacobi operator $J$ such that
\begin{equation*}
q_j = -\frac{k_{j+1}+k_j}{m_j}\,, \qquad
b_j=\frac{k_{j+1}}{\sqrt{m_j m_{j+1}}}\,,
\qquad j\in\mathbb{N}
\end{equation*}
(see \cite{MR2102477,mono-marchenko} for an explanation of the
deduction of these formulae in the finite case). Alternatively, the
system in Fig.~1 can be interpreted as a one dimensional harmonic
crystal \cite[Sec.\,1.5]{MR1711536}. The modified mass-spring system
corresponding to the perturbed operator $\widetilde{J}_n$ is obtained
by adding $\Delta m=m_n(\theta^{-2}-1)$ to the $n$-th mass and $\Delta
k=-hm_n$ to the $n$-th spring constant (see Fig. 2).
\begin{figure}[h]
\begin{center}
\begin{tikzpicture}
  [mass1/.style={rectangle,draw=black!80,fill=black!13,thick,inner sep=0pt,
   minimum size=10mm},
   mass2/.style={rectangle,draw=black!80,fill=black!13,thick,inner sep=0pt,
   minimum size=5.3mm},
   mass3/.style={rectangle,draw=black!80,fill=black!13,thick,inner sep=0pt,
   minimum size=9mm},
   dmass/.style={rectangle,draw=black!80,fill=black!13,thick,inner sep=0pt,
   minimum size=4.8mm},
   wall/.style={postaction={draw,decorate,decoration={border,angle=-45,
   amplitude=0.3cm,segment length=1.5mm}}}]
  \node (mass3) at (7.5,1) [mass3] {\footnotesize$m_{n+1}$};
  \node (mass2) at (4.65,1) [mass2] {\footnotesize$\,m_n$};
  \node (mass1) at (2.6,1) [mass1] {\footnotesize$m_{n-1}$};
  \node (dmass) at (4.65,1.5) [dmass] {\scriptsize$\,\Delta m\,$};
\draw[line width=.8pt,loosely dotted] (0,1) -- (0.4,1);
\draw[decorate,decoration={coil,aspect=0.4,segment
  length=1.3mm,amplitude=1.2mm}] (3.1,1.5) -- node[above=4pt]
{\footnotesize$\Delta k$} (dmass);
\draw[decorate,decoration={coil,aspect=0.4,segment
  length=2.1mm,amplitude=1.8mm}] (0.4,1) -- node[below=4pt]
{\footnotesize$k_{n-1}$}  (mass1);
\draw[decorate,decoration={coil,aspect=0.4,segment
  length=1.5mm,amplitude=1.8mm}] (mass1) -- node[below=4pt]
{\footnotesize$k_{n}$} (mass2);
\draw[decorate,decoration={coil,aspect=0.4,segment
  length=2.5mm,amplitude=1.8mm}] (mass2) -- node[below=4pt]
{\footnotesize$k_{n+1}$} (mass3);
\draw[decorate,decoration={coil,aspect=0.4,segment
  length=2.1mm,amplitude=1.8mm}] (mass3) -- node[below=4pt]
{\footnotesize$k_{n+2}$} (9.7,1);
\draw[line width=.8pt,loosely dotted] (9.8,1) -- (10.2,1);
\end{tikzpicture}
\end{center}
\caption{Perturbed semi-infinite mass-spring system ($n\ge 2$)}\label{fig:2}
\end{figure}
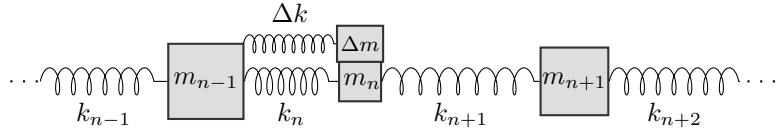

This work carries out a comparative spectral analysis of the operators
$J$ and $\widetilde{J}_n$. This analysis has various distinctive
features related to the kind of perturbation under consideration
(\ref{eq:def-tilde-j}). As mentioned above, the perturbation has a
physical motivation and could be of interest in some
applications. An interesting aspect of the perturbation considered
here is that the comparative spectral analysis of $J$ and
$\widetilde{J}_n$ is susceptible of being treated by a method that
involves the use of quotients of Green functions (see
(\ref{eq:M-definition})) for deriving a master equation (see
(\ref{eq:master})). This method yields results that cannot be obtained
by classical perturbation methods and, remarkably, there is no need of
any general assumption on the spectrum of $J$. In particular, there is
no need to assume that $J$ has discrete spectrum.

It is worth remarking that the perturbation given in
(\ref{eq:def-tilde-j}) has not been studied for the case of
semi-infinite Jacobi matrices. The modification of the spectrum of a
Jacobi operator as a result of a rank-one perturbation is well
understood and has been amply studied (see \cite{MR2305710} and
references therein), however there is scarce
literature treating other kinds of finite rank perturbations.

The main results of this note (Theorems~\ref{thm:interlacing},
\ref{thm:interlacing-infinity}, and \ref{thm:gamma-in-both}) contain
as a particular case all previously known results on the matter
(\cite[Thm.\,2]{MR2915295} and
\cite[Props.\,3.1,\,3.2]{MR2998707}). We point out that the techniques
and ideas developed in this work allow to tackle the corresponding
generalizations of the inverse spectral analysis carried out in
\cite{MR2915295} and \cite{MR2998707}. This is the subject of a
forthcoming paper.

\section{Green functions for Jacobi operators}
\label{sec:green-functions}

Let us consider the following system of difference equations
\begin{subequations}
  \label{eq:spectral-equation}
\begin{align}
  \label{eq:initial-spectral}
   q_1 f_1 + b_1 f_2&=zf_1\\
  \label{eq:recurrence-spectral}
  b_{k-1}f_{k-1} + q_k f_k + b_kf_{k+1}&=zf_k
  \quad k \in \mathbb{N} \setminus \{1\},
\end{align}
\end{subequations}
Clearly, by setting $f_1=1$, the solution of
(\ref{eq:spectral-equation}) can be found uniquely by recurrence. This
solution is an infinite sequence that will be denoted by
$\pi(z)$. Note that $\pi_k(z)$ is a polynomial of degree
$k-1$. Alongside this sequence, we define the sequence $\theta(z)$
as the solution of (\ref{eq:recurrence-spectral}) after
setting $f_1=0$ and $f_2=b_1^{-1}$. Thus, $\theta_k(z)$ is a
polynomial of degree $k-2$. The elements of the sequence
$\pi(z)$, respectively $\theta(z)$, are referred to as the polynomials of
the first, respectively second, kind associated with the matrix
(\ref{eq:jm-0}). By comparing (\ref{eq:main-recurrence}) with
(\ref{eq:spectral-equation}), one concludes that
$\pi(z)\in\ker(J_0^*-zI)$ if and only if $\pi(z)$ is an element of
$l_2(\nats)$. Of course, in particular, $\pi(z)\in\ker(J-zI)$, if and
only if $\pi(z)\in\dom(J)$.

It is easy to verify, directly from the definition of the operator $J$
(see (\ref{eq:main-recurrence})), that
\begin{equation}
  \label{eq:delta-k-through-delta-1}
  \delta_k=\pi_k(J)\delta_1\quad\forall k\in\nats\,.
\end{equation}
This implies that $J$ is simple and $\delta_1$ is a cyclic vector (see
\cite[Sec. 69]{MR1255973}). Therefore, if one defines the spectral
function as
\begin{equation}
  \label{eq:spectral-function-def}
  \rho(t):=\inner{\delta_1}{E(t)\delta_1}\,,\quad t\in\reals\,,
\end{equation}
where $E$ is the resolution of the identity given by the spectral
theorem, then, by \cite[Sec. 69, Thm. 2]{MR1255973}), one has a unitary
map $\Phi:L_2(\reals,\rho)\to l_2(\nats)$ such that $\Phi^{-1}J\Phi$
is the multiplication by the independent variable defined in its
maximal domain. This is the canonical representation of $J$. We note
that, on the basis of \cite[Sec. 69, Thm. 2]{MR1255973}), it follows from
(\ref{eq:delta-k-through-delta-1}) that $\pi_k\in L_2(\reals,\rho)$
for all $k\in\nats$, that is, all moments of $\rho$ exists (see
also \cite[Thm.\,4.1.3]{MR0184042}), and
\begin{equation*}
  \Phi\pi_k=\delta_k\quad\forall k\in\nats\,.
\end{equation*}

In what follows, $\sigma(J)$, $\sigma_p(J)$, and $\sigma_{ess}(J)$ denote
the spectrum, the point spectrum (eigenvalues), and the essential
spectrum (in this case, accumulation points of $\sigma(J)$) of $J$,
respectively.

Now, consider the Weyl $m$-function, given by
\begin{equation}
  \label{eq:weyl-function}
  m(z):=\inner{\delta_1}{(J-z I)^{-1}\delta_1}\,,\qquad z\not\in\sigma(J)\,.
\end{equation}

By using the canonical representation, it immediately follows from the
definition that
\begin{equation*}
  m(z)=\int_{\mathbb{R}}
  \frac{d\rho(t)}{t-z}\,.
\end{equation*}
Thus, by the Nevanlinna representation theorem (see
\cite[Thm.\,5.3]{MR1307384}), $m(z)$ is a Herglotz function.

Due to the inverse Stieltjes transform, one uniquely recovers
$\rho$ from $m$, so $\rho$ and $m$ are in one-to-one correspondence.

For every $z\in\complex\setminus\sigma(J)$, let us consider the element
$\psi(z)$ in $l_2(\nats)$ defined by
\begin{equation}
  \label{eq:psi-vector}
  \psi(z):=(J-zI)^{-1}\delta_1\,.
\end{equation}
It is known \cite[Chap.\,7 Eq.\,1.39]{MR0222718} that for every
$z\in\complex\setminus\sigma(J)$ there exists a unique complex number
$m(z)$ such that
\begin{equation}
  \label{eq:weyl-solution}
  \psi(z)=m(z)\pi(z) +\theta(z)\,.
\end{equation}
The overlap with (\ref{eq:weyl-function}) in the notation is not an
accident, the number $m(z)$ is actually the value of the Weyl
$m$-function at $z$.
\begin{definition}
  \label{def:submatrices}
  For a subspace $\mathcal{G}\subset\l_2(\nats)$, let
  $P_{\mathcal{G}}$ be the orthogonal projection onto
  $\mathcal{G}$. Also, define
  $\mathcal{G}^\perp:=\{\phi\in
  l_2(\nats):\inner{\phi}{\psi}=0\,\forall \psi\in\mathcal{G}\}$ and
  the subspace
$\mathcal{F}_n:=\Span\{\delta_k\}_{k=1}^n$.
  For the operator $J$ given in the Introduction, consider the
  operators
  \begin{equation*}
    J_n^+:=P_{\mathcal{F}_n^\perp}J\eval{\mathcal{F}_n^\perp}\,,\qquad
    J_n^-:=P_{\mathcal{F}_{n-1}}J\eval{\mathcal{F}_{n-1}}
  \end{equation*}
 for any $n\in\nats\setminus\{1\}$. Here, we have used the notation
 $J\eval{\mathcal{G}}$ for the restriction of $J$ to the set
 $\mathcal{G}$, that is,
 $\dom(J\eval{\mathcal{G}})=\dom(J)\cap\mathcal{G}$. Consider also the corresponding $m$-Weyl functions
 \begin{equation*}
   m_n^+(z):=\inner{\delta_{n+1}}{(J_n^+-zI)^{-1}\delta_{n+1}}\,,\qquad
    m_n^-(z):=\inner{\delta_{n-1}}{(J_n^--zI)^{-1}\delta_{n-1}}\,.
 \end{equation*}
\end{definition}
Note that $J_n^+$ is a selfadjoint extension of the operator whose
matrix representation with respect to the basis
$\{\delta_k\}_{k=n+1}^\infty$ of the space
$\mathcal{F}_n^\perp$ is the matrix
(\ref{eq:jm-0}) with the first $n$ rows and $n$ columns
removed. Moreover, when $J_0$ is not essentially selfadjoint, $J_n^+$
has the same boundary conditions at infinity as the operator $J$. Note
that $J_n^-$ is an operator in an $n-1$-dimensional space whose matrix
representation consists of the first $n-1$ columns and $n-1$ rows of
the matrix~(\ref{eq:jm-0}).

The following result can be found in \cite[Eqs. 2.10,
2.16]{MR1616422}. An alternative proof is
provided below.
\begin{proposition}
  \label{prop:m-plus-minus-formulae}
For any $n\in\nats\setminus\{1\}$, one has
  \begin{equation*}
     m_n^+(z)=-\frac{\psi_{n+1}(z)}{b_n\psi_n(z)}\,,\qquad
     m_n^-(z)=-\frac{\pi_{n-1}(z)}{b_{n-1}\pi_n(z)}\,.
  \end{equation*}
\end{proposition}
\begin{proof}
  Define
  \begin{equation*}
    \psi^{(n)}(z):=(J_n^+-zI)^{-1}\delta_{n+1}\,.
  \end{equation*}
Then one verifies that
\begin{equation*}
  \psi_k^{(1)}=-\frac{\psi_{k+1}}{b_1\psi_1}
\end{equation*}
and, after further computations, that
\begin{equation*}
   \psi_k^{(n)}=-\frac{\psi_{k+1}^{(n-1)}}{b_n\psi_1^{(n-1)}}\,.
\end{equation*}
Therefore,
\begin{equation*}
  \psi_k^{(n)}=-\frac{\psi_{k+1}^{(n-1)}}{b_n\psi_1^{(n-1)}}=
  -\frac{\psi_{k+2}^{(n-2)}}{b_n\psi_2^{(n-2)}}=\dots=
  \frac{\psi_{k+n}}{b_n\psi_n}\,.
\end{equation*}
From this, since $m_n^+(z)=\psi_1^{(n)}$, the first identity of the
assertion follows.  The second formula is the statement of
\cite[Lem.\,2]{MR2915295} and a proof for this is provided there.
\end{proof}
It immediately follows from the previous proposition that the
following holds
\begin{corollary}
  \label{cor:pi-k-zero-eigenvalue-j-}
  Fix $n\in\nats\setminus\{1\}$. The real number $x$ is a zero of the
  polynomial $\pi_n(\cdot)$ if and only if $x$ is an eigenvalue of
  $J_n^-$.
\end{corollary}

There are various formulae for the matrix entries of the matrix
representation of $(J-z)^{-1}$ with respect to the canonical
basis. The one provided below is suitable for us
(cf. \cite[Eq.\,2.8]{MR1616422}).
\begin{proposition}
  \label{prop:resolvent-matrix}
  For any $z\in\complex\setminus\sigma(J)$ and $j,k\in\nats$, the
  following holds.
  \begin{equation}
    \label{eq:resolvent-matrix-entries}
    \inner{\delta_j}{(J-zI)^{-1}\delta_k}=
    \begin{cases}
      \pi_j(z)\psi_k(z) & \text{if } j\le k\,,\\
      \psi_j(z)\pi_k(z) & \text{otherwise.}
    \end{cases}
  \end{equation}
\end{proposition}
\begin{proof}
  Fix any $n\in\nats$ and consider the sequence
  \begin{equation*}
    \eta(n,z):=\{\pi_j(z)\psi_n(z)\}_{j=1}^n\cup
    \{\psi_j(z)\pi_n(z)\}_{j=n+1}^\infty\,,
  \end{equation*}
which clearly is in $l_2(\nats)$. By
  the definition of $\pi(z)$ and $\psi(z)$, one verifies that
  \begin{equation*}
    (J-zI)\eta(n,z)=\delta_n
  \end{equation*}
This completes the proof.
\end{proof}

Let us use the following notation
\begin{equation}
  \label{eq:g-k-def}
  G(z,n):=\inner{\delta_n}{(J-zI)^{-1}\delta_n}\,,
\qquad z\in\complex\setminus\sigma(J)
\end{equation}
Note that $m(z)=G(z,1)$ and, in view of
(\ref{eq:delta-k-through-delta-1}) and (\ref{eq:g-k-def}), one has
\begin{equation}
  \label{eq:green-integral}
  G(z,n)=\int_\reals\frac{d\rho_n(t)}{t-z}\,,
\end{equation}
where
\begin{equation}
  \label{eq:rho-def}
  d\rho_n(t):=\pi_n^2(t)d\rho(t)\,.
\end{equation}

The next assertion is found in \cite[Thm. 2.8]{MR1616422}. We provide a simple
proof in which the objects defined above are used.
\begin{proposition}
  \label{prop:Gkk-formula}
  For any $n\in\nats\setminus\{1\}$
\begin{equation}
  \label{eq:Gkk-formula2}
   G(z,n)=\frac{-1}{b_n^2m_n^+(z)+b_{n-1}^2m_n^-(z)+z-q_n}
\end{equation}
\end{proposition}
\begin{proof}
Consider the (modified) Wronskian of the difference
equation~(\ref{eq:recurrence-spectral}) for the sequences $\pi(z)$ and
$\psi(z)$:
  \begin{equation*}
    W_n[\pi(z),\psi(z)]:=
b_n\left(\pi_n(z)\psi_{n+1}(z)-\pi_{n+1}(z)\psi_{n}(z)\right)\,.
  \end{equation*}
Since  $\pi(z)$ and
$\psi(z)$ are solutions of (\ref{eq:recurrence-spectral}), one has
that, for all $n\in\nats$ and $z\in\complex$,
  \begin{equation*}
    W_n[\pi(z),\psi(z)] = 1\,.
  \end{equation*}
Using this and  Proposition~\ref{prop:resolvent-matrix}, one writes
\begin{align*}
  G(z,n)&=\frac{\pi_n\psi_n}{W_n[\pi,\psi]}\\
  &= \frac{1}{b_n\frac{\psi_{n+1}}{\psi_n}-b_n\frac{\pi_{n+1}}{\pi_n}}\,.
\end{align*}
Thus, Proposition~\ref{prop:m-plus-minus-formulae} implies
  \begin{equation}
    \label{eq:Gkk-formula1}
    G(z,n)=\frac{-1}{b_n^2m_n^+(z)-\frac{1}{m_{n+1}^-(z)}}\,.
  \end{equation}
Finally, by means of the formula
\begin{equation*}
  b_{n-1}^2m_n^-(z)+\frac{1}{m_{n+1}^-(z)}-q_n+z=0\,,
\end{equation*}
which follows from Proposition~\ref{prop:m-plus-minus-formulae} and
the definition of $\pi(z)$, one can rewrite (\ref{eq:Gkk-formula1}) as
(\ref{eq:Gkk-formula2}).
\end{proof}

\begin{lemma}
  \label{lem:residue}
  Fix $n\in\nats$ and let $x\not\in\sigma_{ess}(J)$. Then
  \begin{equation}
    \label{eq:res-for}
    \lim_{z\to x}(x-z)G(z,n)=\pi_n^2(x)\rho(\{x\})
  \end{equation}
\end{lemma}
\begin{proof}
  The proof follows from the integral representation of the function
  $G(z,n)$. Indeed,
  \begin{equation}
      \label{eq:nevanlinna-modified}
    G(z,n)=\int_\reals \frac{d\rho_n(t)}{t-z}=
\frac{\pi_n^2(x)\rho(\{x\})}{x-z} +C(z)\,,
  \end{equation}
where $C(z)$ is uniformly bounded inside a closed disk
intersecting $\sigma(J)$ only at $x$. Thus,
\begin{equation*}
  (x-z)G(z,n)=\pi_n^2(x)\rho(\{x\}) +(x-z)C(z)\,.
\end{equation*}
\end{proof}
\begin{remark}
  \label{rem:teschl}
  If $x\in\sigma_p(J)$, Lemma~\ref{lem:residue} is a special case of
  \cite[Eq.\,2.36]{MR1711536}. This result amounts to the
  well-known fact that the residue of the resolvent at an isolated
  eigenvalue is equal to the kernel of the projection onto the
  eigenspace. Indeed, when $x$ is an eigenvalue, taking into account
  that
  \begin{equation*}
    P(x):=\frac{\inner{\pi(x)}{\cdot}}{\norm{\pi(x)}^2}\pi(x)
  \end{equation*}
is a projection onto the corresponding eigenspace, one verifies that
the r.\,h.\,s. of (\ref{eq:res-for}) is the $n$-th diagonal element of
the matrix representation of $P(x)$ with respect to the canonical basis.
\end{remark}
\begin{lemma}
  \label{lem:eigenvalue-zero-or-pole}
  Fix $n\in\nats$ and let $x$ be an isolated eigenvalue of $J$. Then,
  $x$ is a zero of the polynomial $\pi_n(\cdot)$ if and only if $x$ is a zero
  of $G(\cdot,n)$. Also, $x$ is not a zero of $\pi_n(\cdot)$, if and only if
  $x$ is a pole of $G(\cdot,n)$.
\end{lemma}
\begin{proof}
  First we prove that $\pi_n(x)=0$ implies $G(x,n)=0$. According to
  (\ref{eq:resolvent-matrix-entries}) and (\ref{eq:weyl-solution}),
  one has
  \begin{align*}
    G(z,n)&=\pi_n(z)\psi_n(z)\\
    &=m(z)\pi^2_n(z)+\theta_n(z)\pi_n(z)\,.
  \end{align*}
 Since, as in the proof of Lemma~\ref{lem:residue},
 \begin{equation}
   \label{eq:aux-int-representation}
   m(z)=\int_\reals\frac{d\rho(t)}{t-z}=\frac{\rho(\{x\})}{x-z} +C(z)\,,
 \end{equation}
where $C(z)$ is uniformly bounded inside a closed disk
intersecting $\sigma(J)$ only at $x$, it holds that
\begin{align*}
   G(z,n)&=\left(\frac{\rho(\{x\})}{x-z} +C(z)\right)
\pi^2_n(z)+\theta_n(z)\pi_n(z)\\
&=\left(\frac{\rho(\{x\})}{x-z} +C(z)\right)
(x-z)^2h^2(z)+\theta_n(z)\pi_n(z)\,.
\end{align*}
Here we have written $\pi_n(z)=(x-z)h(z)$. The assertion follows
by noticing that
\begin{equation*}
  \lim_{z\to x}\left[(x-z)\rho(\{x\})h^2(z)
    +C(z)(x-z)^2h^2(z)+\theta_n(z)\pi_n(z)
    \right]=0\,.
\end{equation*}
Now let us show that $\pi_n(x)\ne 0$ implies that $\lim_{z\to
  x}G(z,n)=\infty$. Since $\rho(\{x\})\ne 0$, it follows from
Lemma~\ref{lem:residue} that
\begin{equation*}
  \pi_n(x)\ne 0\Rightarrow\lim_{z\to x}(x-z)G(z,n)\ne 0\,.
\end{equation*}
The remaining converse implications follow from
the ones just proven.
\end{proof}
The next assertion is reminiscent of
Corollary~\ref{cor:pi-k-zero-eigenvalue-j-}.
\begin{lemma}
  \label{lem:pi-k-zero-eigenvalue-j-+}
  Let $n\in\nats\setminus\{1\}$. If $x$ is an eigenvalue of
  $J$ and a zero of the polynomial $\pi_n(\cdot)$, then $x$ is an
  eigenvalue of $J_n^+$. On the other hand, if $x$ is an isolated
  eigenvalue of $J$ and $J_n^+$, then $x$ is a zero
  of $\pi_n(\cdot)$.
\end{lemma}
\begin{proof}
  The first assertion is proven by noting that the sequence
  $\{\pi_k(x)\}_{k=n+1}^\infty$ is an eigenvector of $J_n^+$. The second
  assertion is proven by \textit{reductio ad absurdum}. Indeed, assume
  that $x$ is an isolated
  eigenvalue of $J$ and $J_n^+$, and is not a zero of
  $\pi_n(\cdot)$. Then, by
  Lemma~\ref{lem:eigenvalue-zero-or-pole}, $x$ is a pole of
  $G(\cdot,n)$ and therefore, since $x$ is a pole of $m_n^+(\cdot)$,
  Proposition \ref{prop:Gkk-formula} implies that $x$ is also a pole
  of $m_n^-(\cdot)$, that is, an eigenvalue of
  $J_n^-$. Corollary~\ref{cor:pi-k-zero-eigenvalue-j-} shows that this
  contradicts our assumptions.
\end{proof}

\begin{remark}
 \label{rem:rephrase}
 By means of Corollary \ref{cor:pi-k-zero-eigenvalue-j-} one rephrases
 the previous lemma as follows.  For any $n\in\nats\setminus\{1\}$,
  \begin{align*}
    \sigma_p(J)\cap\sigma(J_n^-)&\subset\sigma_p(J_n^+)\\
    \sigma_{disc}(J)\cap\sigma_{disc}(J_n^+)&\subset\sigma(J_n^-)\,,
  \end{align*}
  where the notation
  $\sigma_{disc}(\cdot):=\sigma(\cdot)\setminus\sigma_{ess}(\cdot)$ has been
  used.
\end{remark}

\section{Comparative spectral analysis of $J$ and $\widetilde{J}_n$}
\label{sec:direct-spectral-analysis-general-case}

Let $\widetilde{\rho}(t)$ be the spectral function (see
(\ref{eq:spectral-function-def})) corresponding to the operator
$\widetilde{J}_n$ defined in (\ref{eq:def-tilde-j}). Also, let
$\widetilde{G}(z,n)$ be the function given by (\ref{eq:g-k-def}) with
$\widetilde{J}_n$ instead of $J$. We emphasize the fact that
the value of the subscript of $\widetilde{J}_n$ coincides with the
value of the second argument of $\widetilde{G}(z,n)$, that is,  for
any $n\in\nats$,
\begin{equation*}
  \widetilde{G}(z,n)=\inner{\delta_n}{(\widetilde{J}_n-zI)^{-1}\delta_n}\,.
\end{equation*}

Define, for $n\in\nats$, the function
\begin{equation}
  \label{eq:M-definition}
  \mathfrak{M}_n(z):=\frac{G(z,n)}{\widetilde{G}(z,n)}
\end{equation}
and the constant
  \begin{equation}
    \label{eq:def-gamma}
    \gamma:=\frac{\theta^2h}{1-\theta^2}\,.
  \end{equation}
\begin{lemma}
  \label{lem:master-formula}
For any $n\in\nats$,
\begin{equation}
  \label{eq:master}
  \mathfrak{M}_n(z)=\theta^2+(1-\theta^2)(\gamma -z)G(z,n)
\end{equation}
and
\begin{equation}
  \label{eq:master2}
  \frac{1}{\mathfrak{M}_n(z)}=\frac{1}{\theta^2}+
\left(1-\frac{1}{\theta^2}\right)(\gamma -z)\widetilde{G}(z,n)
\end{equation}
\end{lemma}
\begin{proof}
It follows from (\ref{eq:Gkk-formula2}) and (\ref{eq:M-definition}) that
\begin{equation*}
  \mathfrak{M}_n(z)=
\frac{\theta^2(b_nm_n^+(z)+b_{n-1}m_n^-(z)+z\theta^{-2}-q_n-h)}
{b_nm_n^+(z)+b_{n-1}m_n^-(z)+z-q_n}
\end{equation*}
from which one verifies (\ref{eq:master}) for
$n\in\nats\setminus\{1\}$. For $n=1$, (\ref{eq:master}) follows from
the Riccati equation (see \cite[Eq.\,2.15]{MR1616422}) and
\cite[Eq.\,2.23]{MR1643529})
\begin{equation*}
  b_nm_1^+(z)=q_n-z-\frac{1}{m(z)}
\end{equation*}
after noticing that, in this case,
\begin{equation*}
   P_{\mathcal{F}_1^\perp}\widetilde{J}_1\eval{\mathcal{F}_1^\perp}=
  P_{\mathcal{F}_1^\perp} J\eval{\mathcal{F}_1^\perp}\qquad (\mathcal{F}_1=\Span{\{\delta_1\}})\,.
\end{equation*}
The proof of (\ref{eq:master2}) is completely analogous.
\end{proof}
Equation (\ref{eq:master}) for the case $n=1$ is
\cite[Eq. 18]{MR2998707}. As in \cite{MR2998707}, this equation, now
for $n\in\nats$, is an
important ingredient of the method used for the comparative spectral
analysis of $J$ and $\widetilde{J}_n$. The first immediate consequence
of (\ref{eq:master}) and (\ref{eq:master2}) is the following assertion
\begin{corollary}
  \label{cor:g-gtilde-zero}
  For any $n\in\nats$, when $z\ne\gamma$,
\begin{equation*}
  G(z,n)=0 \Longleftrightarrow \widetilde{G}(z,n)=0
\end{equation*}
\end{corollary}
\begin{proposition}
  \label{prop:j+-j}
   For any $n\in\nats\setminus\{1\}$,
   \begin{equation*}
     \sigma_p(J_n^+)\cap\sigma(J_n^-)\subset
\sigma_p(J)\cap\sigma_p(\widetilde{J}_n)
   \end{equation*}
\end{proposition}
\begin{proof}
  Denote by $\pi^+(z)$ the sequence of polynomials of the first kind
  associated with the Jacobi operator $J_n^+$. Let
  $\lambda\in\sigma(J_n^-)$, then $\pi_n(\lambda)=0$ by
  Corollary~\ref{cor:pi-k-zero-eigenvalue-j-}. This implies that 
the sequences
$\{\pi_j(\lambda)\}_{j=n+1}^\infty$ and
$\{\pi_j^+(\lambda)\}_{j=1}^\infty$ satisfy the same recurrence
relation, including the initial condition
\begin{equation*}
  q_{n+1}f_1+b_{n+1}f_2=zf_1\,.
\end{equation*}
Since the system of recurrence equations with the initial condition is
uniquely solvable modulo a multiplying constant, one has
  \begin{equation}
    \label{eq:polynomialsj-j+}
    \pi_{j}(\lambda)=C\pi_{j-n}^+(\lambda)\,,\qquad j>n\,.
  \end{equation}
The constant $C$ can be found by noting that
\begin{equation*}
  \pi_{n+1}(\lambda)= -b_{n-1}b_n^{-1}\pi_{n-1}(\lambda)\text{ and }
  \pi^+_1(\lambda)=1.
\end{equation*}
Therefore $C=-b_{n-1}b_n^{-1}\pi_{n-1}(\lambda)$.  Since
$\lambda\in\sigma_p(J_n^+)$, the vector $\pi^+(\lambda)$ is in
$\dom(J_n^+)$. Now, by Definition~\ref{def:submatrices}, equation
(\ref{eq:polynomialsj-j+}) implies that $\pi(\lambda)$ is in
$\dom(J)$, that is, $\pi(\lambda)\in\ker(J-\lambda I)$. Finally,
observe that $J_n^+$ and $J_n^-$ do not depend on the perturbation so
the result just proven holds for $\widetilde{J}_n$.
\end{proof}
\begin{proposition}
  \label{thm:gamma-in-both}
  For any $n\in\nats$
  \begin{equation*}
    \gamma\in\sigma(J)\Longleftrightarrow \gamma\in\sigma(\widetilde{J}_n)
  \end{equation*}
\end{proposition}
\begin{proof}
  Let us prove that
  $\gamma\in\sigma(J)\Rightarrow\gamma\in\sigma(\widetilde{J}_n)$. Since
  $\sigma_{ess}(J)=\sigma_{ess}(\widetilde{J})$, it is sufficient to
  verify that
  $\gamma\in\sigma_{disc}(J)$
  implies $\gamma\in\sigma_{disc}(\widetilde{J}_n)$ (recall the
  notation introduced in Remark~\ref{rem:rephrase}). Since
  $\gamma$ is an isolated eigenvalue of $J$, $\gamma$ is either a zero
  or a pole of $G(z,n)$ as a consequence of
  Lemma~\ref{lem:eigenvalue-zero-or-pole}. If
  $G(\gamma,n)=0$, then Lemma~\ref{lem:eigenvalue-zero-or-pole}
  implies that $\pi_n(\gamma)=0$. Therefore, using
  Lemma~\ref{lem:pi-k-zero-eigenvalue-j-+} one concludes that $\gamma$
  is an isolated eigenvalue of $J_n^+$. Thus, taking into account
  Corollary~\ref{cor:pi-k-zero-eigenvalue-j-} and
  Proposition~\ref{prop:j+-j}, $\gamma$ is an eigenvalue of
  $\widetilde{J}_n$. If $\gamma$ is a pole of $G(z,n)$, then using
  Lemma~\ref{lem:residue} and (\ref{eq:master}) one has
  \begin{equation}
    \label{eq:M-gamma-last}
    \mathfrak{M}_n(\gamma)=\theta^2+(1-\theta^2)\pi_n^2(\gamma)\rho(\{\gamma\})\,.
  \end{equation}
 Thus, since $0<\pi_n^2(\gamma)\rho(\{\gamma\})<1$, the
 equality~(\ref{eq:M-gamma-last}) implies that
 $\mathfrak{M}_n(\gamma)\ne 0$. Then, by (\ref{eq:M-definition}),
 $\widetilde{G}(z,n)$ should have a pole in $\gamma$ which, in turn,
 implies the assertion.
\end{proof}

\begin{lemma}
  \label{lem:common-g-zero}
  If
  $r\in\sigma_p(J)\cap\sigma_p(\widetilde{J}_n)\setminus\left(\sigma_{ess}(J)\cup
  \{\gamma\}\right)  $, then $G(r,n)=0$
\end{lemma}
\begin{proof}
  Note that $r$ is an isolated common eigenvalue.  If
  $\widetilde{G}(r,n)=0$, then $G(r,n)=0$ by
  Corollary~\ref{cor:g-gtilde-zero}. If $\widetilde{G}(r,n)\ne 0$,
  then $\frak{M}(r)\in\reals$. Indeed, due to Lemma~\ref{lem:residue},
  one has
  \begin{equation}
    \label{eq:lim-finite}
    \lim_{z\to
      r}\frak{M}(z)=
\frac{\pi_n^2(r)\rho(\{r\})}{\pi_n^2(r)
\widetilde{\rho}(\{r\})}\,,
  \end{equation}
  where it has been used that the $n$-th polynomial of the first kind
  for $\widetilde{J}_n$ coincides with the one for $J$.  By
  Lemma~\ref{lem:eigenvalue-zero-or-pole}, $\pi_n(r)\ne 0$ and
  $r\in\sigma_p(\widetilde{J}_n)$ implies $\widetilde{\rho}(\{r\})\ne
  0$. Therefore, the denominator in the r.h.s. of
  (\ref{eq:lim-finite}) is different from zero. Now, since
  $\frak{M}(r)$ is finite and $r\ne\gamma$, (\ref{eq:master}) implies
  that $G(r,n)$ is finite. Finally, we recur to
  Lemma~\ref{lem:eigenvalue-zero-or-pole} to conclude that
  $G(r,n)=0$.
  \end{proof}
\begin{lemma}
  \label{lem:isolated-eigenvalue}
  Let $r\ne\gamma$ be an isolated eigenvalue of $J$.  Then
  $\rho_n(\{r\})=0$ if and only if
  $r\in\sigma(\widetilde{J}_n)\cap\sigma(J)$.
\end{lemma}
\begin{proof}
  Suppose that $\rho_n(\{r\})=0$ and $r\in\sigma_p(J)$. Then
  $\pi_n(r)=0$ (see (\ref{eq:rho-def}). By Remark~\ref{rem:rephrase},
  one has that $r\in\sigma_p(J_n^+)$. Now, Proposition~\ref{prop:j+-j}
  implies that $r\in\sigma_p(\widetilde{J})$. Let us show that if 
  $r\in\sigma_p(\widetilde{J}_n)$ then $\rho_n(\{r\})=0$. Since $r$ is an
  isolated common eigenvalue, then by Lemma~\ref{lem:common-g-zero}
  $G(r,n)=0$, Therefore, by
  Lemma~\ref{lem:eigenvalue-zero-or-pole}, $\pi_n(r)=0$ which
  in turn implies $\rho_n(\{r\})=0$
\end{proof}
\begin{lemma}
  \label{lem:m0-jtilde-minusj}
  Assume $r\not\in\sigma_{ess}(J)$. Then,
  $r\in\sigma(\widetilde{J}_n)\setminus \sigma(J)$ implies
  $\mathfrak{M}_n(r)=0$. Conversely, if $r\ne\gamma$ and
  $\mathfrak{M}_n(r)=0$, then $ r\in\sigma(\widetilde{J}_n)\setminus
  \sigma(J)$.
\end{lemma}
\begin{proof}
  We begin by proving the first part. By hypothesis $\rho(\{r\})=0$,
  hence $\pi_n^2(r)\rho(\{r\})=0$. Let us show that
  $\pi_n^2(r)\widetilde{\rho}(\{r\})\ne 0$. Indeed, by
  Lemma~\ref{lem:pi-k-zero-eigenvalue-j-+}, $\pi_n(r)\ne
  0$, and $\widetilde{\rho}(\{r\})\ne 0$ since $r$ is an isolated
  element in $\sigma(\widetilde{J}_n)$. Thus, using
  Lemma~\ref{lem:residue}, one obtains
  \begin{equation*}
    \mathfrak{M}_n(r)=\frac{\rho(\{r\})}{\widetilde{\rho}(\{r\})}=0.
  \end{equation*}
  Let us prove the second part. If $r\ne\gamma$ and
  $\mathfrak{M}_n(r)=0$, then, by (\ref{eq:master}), $G(r,n)$ cannot be
  $0$ nor $\infty$ and, therefore, due to
  Proposition~\ref{lem:eigenvalue-zero-or-pole}
  $r\not\in\sigma(J)$. Now, according to (\ref{eq:M-definition}),
  $\mathfrak{M}_n(r)=0$ and $G(r,n)\in\reals\setminus\{0\}$ implies that
  $\widetilde{G}(r,n)=\infty$. Finally, since
  $r\not\in\sigma_{ess}(\widetilde{J}_n)$, one uses
  (\ref{eq:nevanlinna-modified}) to verify that
  $r\in\sigma(\widetilde{J}_n)$.
\end{proof}
The next theorem describes what happens to the spectrum of $J$ when we
perturb it as indicated in (\ref{eq:def-tilde-j}). It states roughly,
that between $\gamma$ (see (\ref{eq:def-gamma})) and an eigenvalue
of $J$ there is exactly one eigenvalue of the perturbation
$\widetilde{J}_n$ and there may be at most one common eigenvalue of
$J$ and $\widetilde{J}_n$. This joint eigenvalue is closer to $\gamma$
than the eigenvalue of $\widetilde{J}_n$ which is not shared by
$J$. Under the perturbation (\ref{eq:def-tilde-j}), the point $\gamma$
acts as an ``atractor'' of eigenvalues. The precise statement is as
follows:
\begin{theorem}
  \label{thm:interlacing}
  Fix an arbitrary $n\in\nats$.
  Let $\theta<1$ and $a,b$ be in
  $\sigma_p(J)\setminus\sigma_p(\widetilde{J}_n)$. Define
  $\mathcal{A}:=(-\infty,\gamma)\cap(a,b)$ with $\gamma>a$, where
  $\gamma$ is defined in (\ref{eq:def-gamma}). Assume
  \begin{enumerate}[(i)]
  \item $(a,b)\cap\sigma_{ess}(J)=\emptyset$
  \item $(a,b)\cap\sigma_p(J)\setminus\sigma_p(\widetilde{J}_n)=\emptyset$
  \end{enumerate}
  Then there
  exists a unique
  $\mu\in\mathcal{A}\cap\sigma_p(\widetilde{J}_n)\setminus\sigma_p(J)$
  and, if
  $\mathcal{A}\cap\sigma_p(J)\cap\sigma_p(\widetilde{J}_n)\ne\emptyset$,
  there exists at most one
  $\eta\in\mathcal{A}\cap\sigma_p(J)\cap\sigma_p(\widetilde{J}_n)$. Moreover
  \begin{equation}
    \label{eq:aux-inequality}
    \abs{\eta-\gamma}<\abs{\mu-\gamma}\,.
  \end{equation}
 The analogous assertion holds for
 $\mathcal{B}:=(a,b)\cap(\gamma,\infty)$, with $\gamma<b$,
 instead of $\mathcal{A}$.
\end{theorem}
\begin{remark}
  \label{rem:observation}
  Observe that the theorem requires that $a,b$ are in $\sigma_p(J)$,
  but not necessarily in $\sigma_{disc}(J)$.
\end{remark}
\begin{proof}
 Step 1. Non common eigenvalues.

 It follows from (\ref{eq:green-integral}) and
  (\ref{eq:master}) that
\begin{equation}
  \label{eq:aux1-interlacing}
  \frac{\mathfrak{M}_n(z)}{(1-\theta^2)(\gamma-z)}=
  \int_\reals\frac{d\rho_n(t)}{t-z}+
\frac{\theta^2}{(1-\theta^2)(\gamma-z)}=\int_\reals\frac{d\omega(t)}{t-z}\,,
\end{equation}
where
\begin{equation*}
  d\omega(t):=d\rho_n(t)+\frac{\theta^2}{1-\theta^2}dh(t)\quad h(t):=
  \begin{cases}
    0 & t<\gamma\\
    1 & t\ge \gamma
  \end{cases}
\end{equation*}
By Lemma~\ref{lem:isolated-eigenvalue}, one has
\begin{equation}
  \label{eq:rhoA}
  \rho_n(\mathcal{A})=\inner{\delta_n}{E_J(\mathcal{A})\delta_n}=0
\end{equation}
which implies that $\omega(\mathcal{A})=0$.  Now, since
$a\in\sigma_p(J)\setminus\sigma_p(\widetilde{J}_n)$, one concludes that
$\rho(\{a\})\ne0$ and, using again
Lemma~\ref{lem:isolated-eigenvalue}, that
$\pi_n^2(a)\ne 0$. Therefore, taking into account
\begin{equation}
  \label{eq:rhoa}
  \rho_n(\{a\})=\int_{\{a\}}\pi^2_n(t)d\rho(t)=\pi_n^2(a)\rho(\{a\})\ne
  0\,,
\end{equation} one obtains
$\omega(\{a\})\ne0$.
Analogously, $\omega(\{b\})\ne0$. On the other hand $\omega(\{\gamma\})
\ne 0$ by the definition of $\omega$.

Then, according to Corollary~\ref{cor:borel-transform}, the function
\begin{equation*}
  \int_\reals\frac{d\omega(t)}{t-x}\,,\qquad x\in\reals
\end{equation*}
has precisely one zero in $\mathcal{A}$, that is, $\mathfrak{M}_n(z)$
has one zero in $\mathcal{A}$. Thus, from
Lemma~\ref{lem:m0-jtilde-minusj}, it follows that there is exactly one
eigenvalue of $\widetilde{J}_n$ in $\mathcal{A}$ which is not an
eigenvalue of $J$.

Step 2. Common  eigenvalues.

If $\gamma>b$, that is if $\mathcal{A}=(a,b)$ or if $\gamma\in (a,b)$
but $\gamma\not\in \sigma_p(J) $ then (\ref{eq:rhoA}) holds. In the
first case (\ref{eq:rhoA}) follows by Lemma
\ref{lem:isolated-eigenvalue} as above and in the second it follows
taking into account that
 $$\rho_n(\{\gamma \})=\int_{\{\gamma \}}\pi^2_n(t)d\rho(t)=\pi_n^2(\gamma )\rho(\{\gamma \})=0\,.$$
As in (\ref {eq:rhoa}) we have $\rho_n(\{a\})\ne0$ and $\rho_n(\{ b
 \})\ne0$

  Then using Corollary~\ref{cor:borel-transform} we obtain that
$$ G(x,n)=\int_\reals\frac{\pi_n^2(t)d\rho(t)}{t-x}$$
has exactly one zero $\eta \in (a,b)$. By Lemma \ref
{lem:common-g-zero} this is the only point which may be a common
eigenvalue of $J$ and $\tilde J$, different from $\gamma $. Therefore
there is at most one point in
$\mathcal{A}\cap\sigma_p(J)\cap\sigma_p(\widetilde{J}_n)$.

If $\gamma \in (a,b) \cap \sigma (J)$ then $\rho (\{ \gamma \})\ne0 $
and $\rho_n(\{\gamma \})=0$ if and only if $\pi_n(\gamma )=0$ If
$\pi_n(\gamma )=0$ then $\rho_n((a,b))=0$ and there is only one root
of $G(x,n)$ in $(a,b)$, so at most one common eigenvalue of $J$ and
$\tilde J$ in $(a,b)$.  If $\pi_n(\gamma )\ne0$ then $\rho_n(\{\gamma
\})\ne0$. Since $\rho_n((a,\gamma )=0=\rho_n(\gamma ,b)$ applying
Corollary~\ref{cor:borel-transform} we get one zero of $G(.,n)$ in
each interval $(a,\gamma )$ and $(\gamma ,b)$. These zeros are, as
before, the only possibilities of common eigenvalues in these
intervals.

Step 3. Proof of inequality (\ref{eq:aux-inequality})

Let $\mu ,\eta \in \mathcal{A}$ be such that $\mu$ is eigenvalue of
$\tilde J$ and not eigenvalue of $J$ and $\eta$ is a common
eigenvalue.

According to (\ref {eq:aux1-interlacing}) $\mu$ satisfies
\begin{equation}
\label{eq:g-noncommon-eigenvalue}
G(\mu ,n)=\int_\reals\frac{d\rho_n(t)}{t-\mu }=
\frac{\theta^2}{(\theta^2-1)(\gamma-\mu )}<0
\end{equation}
and by Lemma \ref
{lem:common-g-zero} we know
\begin{equation}
\label{eq:g-common-eigenvalue}
G(\eta  ,n)=0
\end{equation}

Since $G(.,n)$ is strictly increasing in $\mathcal{A}$ we get $\mu
<\eta $ and therefore
$$|\mu -\gamma |>|\eta -\gamma |$$
 A similar proof works for $\mu ,\eta \in \mathcal{B}$.


\end{proof}
The next theorem considers the situation when $\gamma$ is below the
bottom of the spectrum of $J$. It happens that between $\gamma$ and an
eigenvalue of $J$ there is exactly one eigenvalue of the perturbed
operator $\widetilde{J}_n$. In this interval there are no common
eigenvalues if $\gamma$ is not an eigenvalue.  To the left of $\gamma$
there is no spectrum.
\begin{theorem}
  \label{thm:interlacing-infinity}
  Fix an arbitrary $n\in\nats$.  Let $\theta<1$, $a=-\infty$, and $b$
  be in $\sigma_p(J)\setminus\sigma_p(\widetilde{J}_n)$. Consider
  $\mathcal{A},\mathcal{B}$ as were defined in
  Theorem~\ref{thm:interlacing} and assume the conditions (i), (ii) of
  that theorem. Then there is no eigenvalue of $\widetilde{J}_n$ in
  $\mathcal{A}$. If $\mathcal{B}\ne\emptyset$, then
  $\mathcal{B}\cap\sigma_p(\widetilde{J}_n)\setminus\sigma_p(J)$ has
  precisely one element. If $\gamma\not\in\sigma_p(J)$, then the set
  $\sigma_p(\widetilde{J}_n)\cap\sigma_p(J)\cap (-\infty,b)$ is
  empty. A similar result holds when
  $a\in\sigma_p(J)\setminus\sigma_p(\widetilde{J}_n)$ and $b=\infty$
  just interchanging $\mathcal{A}$ and $\mathcal{B}$.
\end{theorem}
\begin{proof}
  The proof is carried out for the case $\gamma<b$.
  By (\ref{eq:rhoA}) one has that
  $\rho_n(-\infty,\gamma)=0$. Hence
  \begin{equation*}
    G(\lambda,n)=\int_\gamma^\infty\frac{d\rho_n(t)}{t-\lambda}\,.
  \end{equation*}
  For any $\lambda\in\mathcal{A}$, and $t>\gamma$, we have
  $t-\lambda>0$, therefore $G(\lambda,n)>0$ since
  $\rho_n(\gamma,\infty)\ne0$ as a consequence of
  Lemma~\ref{lem:isolated-eigenvalue}. This contradicts
  (\ref{eq:g-noncommon-eigenvalue}) with $\mu=\lambda$. Thus, there is
  no eigenvalue of $\widetilde{J}_n$ which is not an eigenvalue of $J$
  in $\mathcal{A}$. Moreover, $G(\lambda,n)>0$ also contradicts
  (\ref{eq:g-common-eigenvalue}) with $\eta=\lambda$, hence there are
  no common eigenvalues in $\mathcal{A}$.

  Now, let us see what happens in $\mathcal{B}=(\gamma,b)$. We first
  treat the case when $\rho_n(\{\gamma\})=0$. Then,
  Lemma~\ref{lem:isolated-eigenvalue} and (ii) implies that
  $\rho_n(-\infty,b)=0$. By
  Theorem~\ref{thm:analysis-borel-transform}, $G(\lambda,n)$ is
  continuous and strictly increasing. Moreover, reasoning as at the
  beginning of this proof, $G(\lambda,n)$ is positive. On the other hand,
  Theorem~\ref{thm:analysis-borel-transform} implies
  \begin{equation*}
    \lim_{\lambda\uparrow b}G(\lambda,n)=+\infty\,.
  \end{equation*}
  Define
  \begin{equation*}
    f(\lambda):=\frac{\theta^2}{(\theta^2-1)(\gamma-\lambda)}\,.
  \end{equation*}
 By our assumption on $\theta$, $f(\lambda)>0$ for
 $\lambda>\gamma$. Also,
 \begin{equation*}
   \lim_{\lambda\downarrow \gamma}f(\lambda)=+\infty
 \end{equation*}
 and $f$ is decreasing in $(\gamma,b)$.Thus, for $\lambda$ close to
 $\gamma$, one verifies that $G(\lambda,n)< f(\lambda)$ and for
 $\lambda$ close to $b$, that $G(\lambda,n)> f(\lambda)$. Therefore,
 there exists a unique $\lambda_0\in(\gamma,b)$ such that
 $G(\lambda_0,n)=f(\lambda_0)$. This $\lambda_0$ is an eigenvalue of
 $\widetilde{J}_n$ which is not an eigenvalue of $J$. In $(-\infty,b)$
 there are no common eigenvalues since $G(\lambda,n)>0$ in this
 interval.

 Suppose now that $\rho_n(\{\gamma\})\ne 0$. By hypothesis, it follows
 from Lemma~\ref{lem:isolated-eigenvalue} that $\rho_n(\{b\})\ne
 0$. Thus, it holds that $\rho_n(\{\gamma\})\ne 0$, $\rho_n(\{b\})\ne
 0$, and $\rho_n(\gamma,b)=0$. By
 Theorem~\ref{thm:analysis-borel-transform}, the function
 $G(\lambda,n)$ is continuous, strictly increasing in $(\gamma,b)$ and
 \begin{equation*}
   \lim_{\lambda\downarrow \gamma}G(\lambda,n)=-\infty\qquad
   \lim_{\lambda\uparrow b}G(\lambda,n)=+\infty\,.
 \end{equation*}
 Hence there is a unique $\lambda_0\in(\gamma,b)$ such that
 $G(\lambda_0,n)=f(\lambda_0)$. This $\lambda_0$ is an eigenvalue of
 $\widetilde{J}_n$ which is not an eigenvalue of $J$. Since there is
 also exactly one $\lambda_1$ in $(\gamma,b)$ such that
 $G(\lambda_1,n)=0$, one may have a common eigenvalue in this interval.
\end{proof}
\begin{remark}
  \label{rem:theta-greater}
  When $\theta>1$, the assertions of Theorems~\ref{thm:interlacing}
  and \ref{thm:interlacing-infinity}, modified by interchanging $J$
  and $\widetilde{J}_n$, hold true. The proof is carried out in the
  same way, but using (\ref{eq:master2}) instead of (\ref{eq:master})
  and $\widetilde{G}(z,n)$ instead of $G(z,n)$
\end{remark}
\begin{remark}
  \label{rem:discrete}
  If $n=1$ and $\sigma_{ess}(J)=\emptyset$, then
  Theorems~\ref{thm:interlacing}, \ref{thm:interlacing-infinity}, and
  Proposition~\ref{thm:gamma-in-both} are Propositions~3.1 and 3.2 of
  \cite{MR2998707}. Thus, in this case, a complete description of the
  interplay of the spectra is obtained.
\end{remark}
\begin{remark}
  \label{rem:finite}
  Note that the validity of our results includes the case of finite
  dimensional Jacobi matrices. In this particular case, the
  results of this work coincide with the corresponding ones in
  \cite{MR2915295}.
\end{remark}
\appendix
We give a simple proof of the following known results for Nevanlinna
functions.
\begin{theorem}
  \label{thm:analysis-borel-transform}
  Let $\rho$ be a positive measure on $\reals$ such that
  \begin{enumerate}[(i)]
  \item $\rho(\reals)<\infty$
  \item $\rho(a,b)=0$ for an open interval $(a,b)$
  \item If $a\ne -\infty$, then $\rho(\{a\})\ne 0$, and if $b\ne
    +\infty$, then  $\rho(\{b\})\ne 0$.
  \end{enumerate}
  Define, for $\lambda\in\reals$,
  \begin{equation*}
    F(\lambda):=\int_\reals\frac{d\rho(t)}{t-\lambda}\,.
  \end{equation*}
 Then $F$ has the following properties
 \begin{enumerate}[(I)]
 \item $F$ is continuous in $(a,b)$
 \item $F$ is strictly increasing in  $(a,b)$
 \item If $a\ne -\infty$, then
   \begin{equation*}
     \lim_{\lambda\downarrow a}F(\lambda)=-\infty
   \end{equation*}
   If  $b\ne +\infty$, then
   \begin{equation*}
     \lim_{\lambda\uparrow b}F(\lambda)=+\infty
   \end{equation*}
 \end{enumerate}
\end{theorem}
\begin{proof}
  By the definition of $F$, one has
  \begin{align*}
    \frac{F(x+h)-F(x)}{h}&=\int_\reals\frac{1}{h}
    \left(\frac{1}{t-(x+h)}-\frac{1}{t-x}\right)d\rho(t)\\
    &=\int_{\reals\setminus(a,b)}\frac{d\rho(t)}{(t-(x+h))(t-x)}\,.
  \end{align*}
 Let $x\in\mathcal{E}\subset(a,b)$, where $\mathcal{E}$ is a closed
  interval. Then, for $t\in\reals\setminus(a,b)$, there exists $d>0$
  such that $\abs{t-x}>d>0$. Choose $h\in\reals$ such that
  $\abs{h}<\frac{d}{2}$. Then
  \begin{equation*}
    \abs{t-xh}>\abs{t-x}-\abs{h}\ge\frac{d}{2}>0\,.
  \end{equation*}
 Therefore
 \begin{equation*}
   \abs{\frac{1}{(t-x-h)(t-x)}}\le \frac{2}{d^2}\in L^1(\reals,d\rho)
 \end{equation*}
since $\rho$ is finite. Thus, one can apply the dominated convergence
theorem to obtain that
\begin{equation*}
  \lim_{h\to
    0}\frac{F(x+h)-F(x)}{h}=\int_{\reals\setminus(a,b)}\frac{d\rho(t)}{(t-x)^2}>0.
\end{equation*}
This proves (I) and (II).

Now, let $\{b_n\}_{n=1}^\infty$ be a nondecreasing real sequence such
that $b_n\convergesto{n} b$. Then, for $n$ sufficiently large,
\begin{equation*}
  F(b_n)=\left(\int_{(-\infty,a]}+\int_{[b,+\infty)}\right)\frac{d\rho(t)}{t-b_n}\,.
\end{equation*}
For the last term of the r.\,h.\,s., one has
\begin{equation}
  \label{eq:appendix-first-int}
  \int_{[b,+\infty)}\frac{d\rho(t)}{t-b_n}=\frac{\rho(\{b\})}{b-b_n}+\int_{(b,+\infty)}\frac{d\rho(t)}{t-b_n}\convergesto{n}\infty\,.
\end{equation}
On the other hand, for $n$ sufficiently large
\begin{equation}
 \label{eq:appendix-second-int}
  \abs{\int_{(-\infty,a]}\frac{d\rho(t)}{t-b_n}}\le
  \int_{(-\infty,a]}\frac{d\rho(t)}{\abs{t-b_n}}
 \le C\int_\reals d\rho(t)<\infty
\end{equation}
since $\{b_n\}_{n=1}^\infty$ accumulates at $b$ and $t\in(-\infty,a]$.

It follows from (\ref{eq:appendix-first-int}) and
(\ref{eq:appendix-second-int}) that $F(b_n)\convergesto{n}\infty$ and
therefore $F(\lambda)$ tends to $+\infty$ whenever $\lambda\to b$. A similar
argument proves that $F(\lambda)\to -\infty$ if $\lambda\to a$.
\end{proof}
\begin{remark}
  \label{rem:teschl2}
 $F$ is not only continuous but holomorphic away of the
  support $\rho$ (see the paragraph after \cite[Lem.\,B.4]{MR1711536}).
  \end{remark}
\begin{corollary}
  \label{cor:borel-transform}
  Let $F$ be as in Theorem~\ref{thm:analysis-borel-transform}. If
  $a\ne-\infty$ and $b\ne+\infty$, then there
  exists exactly one point $p\in(a,b)$ such that $F(p)=0$.
\end{corollary}
\begin{proof}
  The proof follows directly from (I), (II), (III) of
  Theorem~\ref{thm:analysis-borel-transform}.
\end{proof}
\begin{acknowledgments}
  We thank the anonymous referees for useful suggestions and comments
  that led in particular to Remarks \ref{rem:teschl} and
  \ref{rem:teschl2}.
\end{acknowledgments}
\def\cprime{$'$} \def\lfhook#1{\setbox0=\hbox{#1}{\ooalign{\hidewidth
  \lower1.5ex\hbox{'}\hidewidth\crcr\unhbox0}}} \def\cprime{$'$}

\end{document}